\documentclass[12pt,onecolumn,journal,draftclsnofoot]{IEEEtran}

\usepackage[utf8x]{inputenc}
\usepackage{amsmath,amssymb,mathtools}
\usepackage{amsthm}
\usepackage{color}
\usepackage{graphicx}
\usepackage{cite}
\usepackage{amsfonts}
\usepackage{microtype}
\usepackage{tikz}
\usetikzlibrary{calc,arrows,shapes.geometric}

\ifCLASSOPTIONdraftcls
\title{The Optimal Input Distribution for\\
Partial Decode-and-Forward in the\\
MIMO Relay Channel}
\else
\title{The Optimal Input Distribution for Partial Decode-and-Forward in the MIMO Relay Channel}
\fi

\ifCLASSOPTIONdraftcls
\author{Lennart Gerdes, Christoph Hellings, Lorenz Weiland, and Wolfgang Utschick%
\thanks{%
This work was supported by the Deutsche Forschungsgemeinschaft (DFG) under grant Ut36/11.
The authors are with the \red{Fachgebiet Methoden der Signalverarbeitung}, 
Technische Universit\"at M\"unchen, 80290 M\"unchen, Germany; 
e-mail: \{gerdes, hellings, lorenz.weiland, utschick\}@tum.de.}}
\else
\author{%
Lennart Gerdes,~\IEEEmembership{Student Member,~IEEE}, 
Christoph Hellings,~\IEEEmembership{Student Member,~IEEE},\\ 
Lorenz Weiland, and 
Wolfgang Utschick,~\IEEEmembership{Senior Member,~IEEE}%
\thanks{%
This work was supported by the Deutsche Forschungsgemeinschaft (DFG) under grant Ut36/11.
The authors are with the \red{Fachgebiet Methoden der Signalverarbeitung},
Technische Universit\"at M\"unchen, 80290 M\"unchen, Germany; 
e-mail: \{gerdes, hellings, lorenz.weiland, utschick\}@tum.de.}}
\fi

\DeclareMathAlphabet{\mathbit}{OML}{cmr}{bx}{it}

\DeclareMathOperator{\E}{E}
\DeclareMathOperator{\rank}{rank}

\newcommand{\red}{}

\newcommand{\ntilde}[1]{\expandafter\tilde#1}
\newcommand{\nbar}[1]{\expandafter\bar#1}

\newcommand{\x}{\mathbit{x}}
\newcommand{\y}{\mathbit{y}}
\newcommand{\z}{\mathbit{z}}
\newcommand{\xs}{\mathbit{x}_{\textnormal{S}}}
\newcommand{\xr}{\mathbit{x}_{\textnormal{R}}}
\newcommand{\yr}{{\mathbit{y}}_{\textnormal{R}}}
\newcommand{\yd}{{\mathbit{y}}_{\textnormal{D}}}
\renewcommand{\u}{\mathbit{u}}
\renewcommand{\v}{\mathbit{v}}
\newcommand{\q}{\mathbit{q}}

\newcommand{\nr}{{\mathbit{n}}_{\textnormal{R}}}
\newcommand{\nd}{\mathbit{n}_{\textnormal{D}}}
\newcommand{\Zr}{{\mathbit{Z}}_{\textnormal{R}}}
\newcommand{\Zd}{\mathbit{Z}_{\textnormal{D}}}
\newcommand{\Z}{{\mathbit{Z}}}
\newcommand{\Zrt}{\ntilde{\Zr}}

\newcommand{\Ha}{\mathbit{H}}
\newcommand{\Hsr}{{\mathbit{H}}_{\textnormal{SR}}}
\newcommand{\Hsd}{{\mathbit{H}}_{\textnormal{SD}}}
\newcommand{\Hrd}{\mathbit{H}_{\textnormal{RD}}}
\newcommand{\HSRd}{[\Hsd, \Hrd]}

\newcommand{\Usr}{\mathbit{U}_{\textnormal{SR}}}
\newcommand{\Usd}{\mathbit{U}_{\textnormal{SD}}}
\newcommand{\Vsr}{\mathbit{V}_{\textnormal{SR}}}
\newcommand{\Vsd}{\mathbit{V}_{\textnormal{SD}}}
\newcommand{\Ssr}{\mathbit{\Sigma}_{\textnormal{SR}}}
\newcommand{\Ssd}{\mathbit{\Sigma}_{\textnormal{SD}}}

\newcommand{\Bsr}{\mathbit{B}_{\textnormal{SR}}}
\newcommand{\Bsd}{\mathbit{B}_{\textnormal{SD}}}

\newcommand{\Nd}{N_{\textnormal{D}}}
\newcommand{\Ns}{N_{\textnormal{S}}}
\newcommand{\Nr}{N_{\textnormal{R}}}

\newcommand{\setC}{\mathbb{C}}
\newcommand{\setR}{\mathbb{R}}
\newcommand{\CN}{\mathcal{N}_{\mathbb{C}}}
\newcommand{\lam}{\mathbit{\Lambda}}
\newcommand{\zeros}{\boldsymbol{0}}
\newcommand{\id}{\mathbf{I}}

\newcommand{\A}{\mathbit{A}}
\renewcommand{\S}{\mathbit{S}}
\newcommand{\C}{\mathbit{C}}
\newcommand{\Cs}{\mathbit{C}_{\textnormal{S}}}
\newcommand{\Cr}{\mathbit{C}_{\textnormal{R}}}
\newcommand{\Cq}{\mathbit{C}_{\textnormal{Q}}}
\newcommand{\Cu}{\mathbit{C}_{\textnormal{U}}}
\newcommand{\Cur}{\mathbit{C}_{\text{UR}}}
\newcommand{\Cugr}{\mathbit{C}_{\text{U}\mid\text{R}}}
\newcommand{\Cv}{\mathbit{C}_{\textnormal{V}}}
\newcommand{\Q}{\mathbit{Q}}

\newcommand{\Ds}{\mathbit{D}_{\text{S}}}
\newcommand{\Dr}{\mathbit{D}_{\text{R}}}
\newcommand{\st}{\text{s.\,t.}}
\newcommand{\he}{\text{H}}

\newcommand{\trace}{\text{tr}}
\newcommand{\PS}{P_{\text{S}}}
\newcommand{\PR}{P_{\text{R}}}

\newcommand{\psd}{\succeq}
\newcommand{\nsd}{\preceq}
\newcommand{\posd}{\succ}

\newcommand{\markov}{\leftrightarrow}

\newcommand{\rdf}{R_{\textnormal{DF}}}
\newcommand{\rpdf}{R_{\textnormal{PDF}}}
\newcommand{\rpdfn}{R_{\textnormal{PDF}}^{\,\CN}}
\newcommand{\rp}{R_{\textnormal{P2P}}}
\newcommand{\csb}{C_{\textnormal{CSB}}}

\newtheorem{theorem}{Theorem}

\newtheorem{definition}{Definition}

\begin{document}
\maketitle



\begin{abstract}
This paper considers the partial decode-and-forward (PDF) strategy for the Gaussian multiple-input multiple-output (MIMO) relay channel. Unlike for the decode-and-forward~(DF) strategy or point-to-point (P2P) transmission, for which Gaussian channel inputs are known to be optimal, the input distribution that maximizes the achievable PDF rate for the Gaussian MIMO relay channel has remained unknown so far. For some special cases, e.g., for relay channels where the optimal PDF strategy reduces to DF or P2P transmission, it could be deduced that Gaussian inputs maximize the PDF rate. For the general case, however, the problem has remained open until now. In this work, we solve this problem by proving that the maximum achievable PDF rate for the Gaussian MIMO relay channel is always attained by Gaussian channel inputs. Our proof relies on the channel enhancement technique, which was originally introduced by Weingarten et al. to derive the (private message) capacity region of the Gaussian MIMO broadcast channel. By combining this technique with a primal decomposition approach, we first establish that jointly Gaussian source and relay inputs maximize the achievable PDF rate for the aligned Gaussian MIMO relay channel. Subsequently, we use a limiting argument to extend this result from the aligned to the general Gaussian MIMO relay channel. 
\end{abstract}


\section{Introduction}\label{sec:introduction}
 
This work considers the Gaussian multiple-input multiple-output~(MIMO) relay channel, a three-node network where one source wants to convey information to one destination with the help of a single relay. All three nodes may be equipped with multiple antennas and they are connected by additive Gaussian noise channels. Furthermore, it is assumed that the relay does not have own information to transmit or receive so that its only purpose is to assist the communication from the source to the destination.

The concept of relaying traces back to van der Meulen~\cite{Meulen71Three-TerminalCommunicationChannels}, who introduced the first information theoretic model for the relay channel. While the capacity of the relay channel is still unknown, substantial advances towards its information theoretic understanding have since been made. The most important work on the relay channel is by Cover and El Gamal~\cite{Cover79CapacityTheoremsRelay}, who derived a capacity upper bound and achievable rates based on a then new \emph{cut-set bound}~(CSB) and two coding schemes that are nowadays referred to as \emph{decode-and-forward}~(DF) and \emph{compress-and-forward}~(CF), respectively. The DF strategy requires the relay to decode the entire source message, which is then re-encoded and, in cooperation with the source, transmitted to the destination. When using CF, the relay reliably forwards an estimate, i.e., a compressed version of its received signal, to the destination. In~\cite{Kramer05CooperativeStrategiesand}, these two basic strategies were generalized to various relay channel models that include multiple sources, relays, or destinations.

In their pioneering work, Cover and El Gamal also proposed a more general coding scheme that combines the DF and CF strategies~\cite[Theorem~7]{Cover79CapacityTheoremsRelay}. If the relay uses this strategy, it decodes only a part of the source message and compresses the remainder. The \emph{partial decode-and-forward} (PDF) scheme is a special case of this mixed strategy where the relay only forwards information about the part of the source message it has decoded. We remark that PDF in turn includes the DF strategy and \emph{point-to-point} (P2P) transmission from source to destination as special cases. Since the PDF scheme allows to optimize the amount of information the relay must decode, it provides the possibility to tradeoff sending information via the relay versus sending it over the direct link. In particular, equipping all nodes with multiple antennas creates spatial degrees of freedom which the PDF scheme may exploit to outperform the DF scheme.

Upper and lower bounds on the capacity of the Gaussian MIMO relay channel were first studied in~\cite{Wang05CapacityofMIMO}, where it was shown that Gaussian channel inputs maximize both the CSB and the achievable DF rate. Furthermore, a generally loose upper bound on the CSB was established and evaluated, and different lower bounds on the capacity based on suboptimal DF strategies or P2P transmission were also derived. Using the fact that Gaussian channel inputs maximize the CSB and the achievable DF rate, it was then independently shown in~\cite{Ng11TransmitSignaland} and~\cite{Gerdes11OptimizedCapacityBoundsa} that, if perfect \emph{channel state information} (CSI) is available at all nodes, the corresponding optimal values can be determined as the solutions of convex optimization problems.

Employing PDF in the Gaussian MIMO relay channel was first considered in~\cite{Lo08RateBoundsMIMO}, where the strategy was termed ``transmit-side message splitting''. The authors formulated the PDF rate maximization problem for jointly Gaussian source and relay inputs, but they did not solve the resulting nonconvex problem. In addition, no attempt was made to characterize the input distribution that maximizes the achievable PDF rate. Rather, proper complex Gaussian channel inputs were assumed as part of the system model. If the channel inputs are restricted to be complex Gaussian, it was then shown in~\cite{Hellings14OptimalGaussianSignaling} that jointly proper source and relay inputs are indeed optimal. For the general case, however, the optimal input distribution has been unknown so far. Consequently, it has not been possible to characterize the maximum achievable PDF rate for the general Gaussian MIMO relay channel.

This is in contrast to the CSB, the DF rate, and the P2P capacity, for which it is well known that Gaussian inputs are optimal, cf.~\cite{Telatar99CapacityofMulti-antenna, Wang05CapacityofMIMO}. The maximum achievable PDF rate for the Gaussian MIMO relay channel can thus be characterized whenever the optimal PDF strategy is equivalent to the DF strategy or P2P transmission, or if PDF achieves the CSB. Such special cases include the (physically) degraded and the reversely degraded relay channel~\cite{Cover79CapacityTheoremsRelay}, the semideterministic relay channel~\cite{ElGamal82CapacityofSemideterministic}, the relay channel with orthogonal components~\cite{ElGamal05CapacityofClass}, as well as the stochastically degraded and the reversely stochastically degraded relay channel~\cite{Gerdes14OptimalPartialDecode-and-Forward}.\footnote{We remark that some of these special cases are only of theoretic interest. In particular, the Gaussian MIMO relay channel is never (physically) degraded or reversely degraded if the relay and destination noise vectors are assumed to be independent, and it is never semideterministic unless the relay does not experience any noise.} Moreover, the achievable PDF rate is also maximized by Gaussian channel inputs if the row spaces of the source-to-relay and the source-to-destination channel gain matrices are disjoint~\cite{Gerdes14OptimalPartialDecode-and-Forwarda}.

In this paper, we generalize these previous results by showing that the maximum achievable PDF rate for the Gaussian MIMO relay channel is always attained by jointly Gaussian source and relay inputs. To this end, we first establish that jointly Gaussian source and relay inputs maximize the achievable PDF rate for the \emph{aligned} Gaussian MIMO relay channel, which constitutes the main challenge of the proof. Subsequently, we use a limiting argument to extend this result from the aligned to the general Gaussian MIMO relay channel. We remark that the idea to first consider an aligned channel was introduced by Weingarten et al.~\cite{Weingarten06CapacityRegionof} to derive the (private message) capacity region of the Gaussian MIMO broadcast channel.

The proof that Gaussian channel inputs maximize the achievable PDF rate for the aligned Gaussian MIMO relay channel requires a large variety of ingredients. For the achievability part, we simply use jointly Gaussian source and relay inputs. The converse is based on a \emph{channel enhancement} argument, which, like the idea to first consider the aligned channel, goes back to~\cite{Weingarten06CapacityRegionof}. More specifically, the enhanced aligned relay channel we consider in the converse is stochastically degraded. As a result, its maximum achievable PDF rate is attained by a pure DF strategy~\cite{Gerdes14OptimalPartialDecode-and-Forward}, for which Gaussian channel inputs are known to be optimal~\cite{Wang05CapacityofMIMO}. Finally, the key to proving that achievability and converse meet is a primal decomposition approach, which we use to split the complicated PDF rate maximization into subproblems. In a slightly different context, this decomposition was first proposed in our previous work~\cite{Hellings14OptimalGaussianSignaling}. Therein, it enabled us to show and exploit the mathematical equivalence between one of the resulting subproblems and a sum rate maximization problem for a Gaussian MIMO broadcast channel with dirty paper coding. For the proof presented in this paper, however, we instead obtain a subproblem that is mathematically equivalent to the problem of finding the secrecy capacity of the aligned Gaussian MIMO wiretap channel (vector Gaussian wiretap channel) under shaping constraints~\cite[Section~II-A]{Liu09NoteSecrecyCapacity}. To facilitate the proof that jointly Gaussian source and relay inputs maximize the achievable PDF rate for the aligned Gaussian MIMO relay channel, we can thus adopt considerations that were used in the derivation of~\cite[Theorem~2]{Liu09NoteSecrecyCapacity}. 

\emph{Notation:} $\setR_{+}$ stands for the set of nonnegative real numbers. Matrices are denoted by bold capital letters, vectors by bold lowercase characters. The identity matrix and the all-zeros matrix/vector are represented by $\id$ and $\zeros$, respectively, where the dimensions are indicated by subscripts if necessary. $\mathbit{A}^{\he}$, $\mathbit{A}^{-1}$, $\mathbit{A}^{+}$, $|\mathbit{A}|$, and $\trace(\mathbit{A})$ denote the conjugate transpose, inverse, Moore-Penrose pseudoinverse, determinant, and trace of matrix $\mathbit{A}$, while $\mathbit{A} \psd \mathbit{B}$ and $\mathbit{A} \posd \mathbit{B}$ mean that $\mathbit{A} - \mathbit{B}$ is positive semidefinite (nonnegative definite) and positive definite, respectively. $\E[\cdot]$ is the expectation operator and $\mathbit{x} \sim \CN(\zeros, \mathbit{C})$ means that $\mathbit{x}$ is a zero-mean proper (circularly symmetric) complex Gaussian random vector with covariance matrix $\mathbit{C}$. Finally, $I(\x;\y|\z)$ is the conditional mutual information of $\x$ and $\y$ given $\z$, and $h(\x|\z)$ denotes the conditional differential entropy of $\x$ given $\z$.

\section{System Model}\label{sec:model}

The channel model for the Gaussian MIMO relay channel, which is illustrated in Figure~\ref{fig:RC}, is obtained by applying the linear MIMO model to the considered relay scenario. The receive signal vectors of the relay and the destination can thus be expressed as
\begin{align}\begin{alignedat}{3}\label{eq:RC}
	\yr &= \Hsr \xs + \nr, \qquad& \nr &\sim \CN(\zeros,\Zr), \\
	\yd &= \Hsd \xs + \Hrd \xr + \nd, \qquad& \nd &\sim \CN(\zeros,\Zd),
\end{alignedat}\end{align}
where $\Hsr \in \setC^{\Nr \times \Ns}$, $\Hsd \in \setC^{\Nd \times \Ns}$, and $\Hrd \in \setC^{\Nd \times \Nr}$ represent the channel gain matrices of appropriate dimensions, which are assumed to be perfectly and instantaneously known at all nodes. Moreover, $\nr \sim \CN(\zeros,\Zr)$ and $\nd \sim \CN(\zeros,\Zd)$ denote zero-mean proper complex Gaussian noise vectors with nonsingular covariance matrices $\Zr \in \setC^{\Nr \times \Nr}$ and $\Zd \in \setC^{\Nd \times \Nd}$. The noise vectors are independent of each other and independent of the transmit signals $\xs \in \setC^{\Ns}$ and $\xr \in \setC^{\Nr}$. Finally, perfectly synchronized transmission and reception between all nodes is assumed, and it is implicit in~\eqref{eq:RC} that the relay operates in full-duplex mode and is able to completely cancel its own self-interference.
\begin{figure}[t]
	\centering
	\begin{tikzpicture}[scale=0.5,>=stealth']
		\tikzset{device/.style={circle,line width=1.25pt,minimum width=1.95em,draw}}                                       
		\node[device] (source) at (0,0) [label={[inner sep=0,yshift=-.5ex]270:$\xs$}] {S};
		\node[device] (relay) at (4.5,2.5) [label={[inner sep=0,yshift=.5ex]90:$\yr \mid \xr$}] {R};
		\node[device] (destination) at (9,0) [label={[inner sep=0,yshift=-.5ex]270:$\yd$}] {D};
		\draw [line width=1.25pt,->] (source) to node [auto] {$\Hsr$} (relay);
		\draw [line width=1.25pt,->] (source) to node [below] {$\Hsd$} (destination);
		\draw [line width=1.25pt,->] (relay) to node [auto] {$\Hrd$} (destination);
	\end{tikzpicture}
	\caption{Illustration of the Gaussian MIMO Relay Channel}
	\label{fig:RC}
\end{figure}
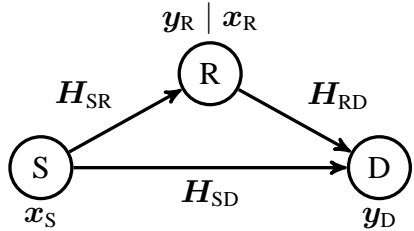

Without further conditions on the channel inputs $\xs$ and $\xr$, the capacity of the Gaussian MIMO relay channel is infinite. That is because one can then choose infinite subsets of inputs arbitrarily far apart so that they are distinguishable at the outputs with arbitrarily small probability of error, cf.~\cite[Chapter~9]{Cover2006}. We therefore impose the transmit power constraints
\begin{align}
	\E[\xs^{\he}\xs] \leq \PS, \qquad \E[\xr^{\he}\xr] \leq \PR
\end{align}
on the channel inputs, where $\PS>0$ and $\PR>0$ denote the power budgets available at the source and the relay, respectively.

Note that, without loss of generality, we can restrict our attention to zero-mean channel inputs as it is clear that the optimal $\xs$ and $\xr$ are always zero-mean. The reason for this is that channel inputs with nonzero mean consume more transmit power than the corresponding zero-mean signals, but they cannot convey more information since translations do not change the differential entropy of continuous random vectors, cf.~\cite[Theorem~8.6.3]{Cover2006}. As a consequence, the covariance matrices of the source and relay inputs are given by $\Cs = \E[\xs \xs^{\he}]$ and $\Cr = \E[\xr \xr^{\he}]$ so that the transmit power constraints can equivalently be expressed as
\begin{align}
	\trace(\Cs) \leq \PS, \qquad \trace(\Cr) \leq \PR.
\end{align}

\section{Partial Decode-and-Forward (PDF)}

The partial decode-and-forward~(PDF) strategy can be viewed as a generalization of the well-known decode-and-forward~(DF) scheme. When using DF, the relay is required to decode the entire message transmitted by the source, even if this means that the source-to-relay link becomes the bottleneck of the communication. One way to overcome this problem is to allow the relay to partially decode the source message. For this purpose, the message $W$ that is to be transmitted from the source to the destination is split into two independent parts $W'$ and $W''$, of which the relay is only required to decode $W'$. By constructing separate codebooks for $W'$ and $W''$ and using superposition coding at the source, a PDF scheme that achieves all rates smaller than or equal to
\begin{multline}\label{eq:pdf-opt}
	\rpdf = \max_{p(\u,\xs,\xr)} \min\, \bigl\{ I(\u; \yr | \xr) + I(\xs; \yd | \u,\xr), I(\xs, \xr; \yd) \bigr\} \\
	\st \quad \u \markov (\xs,\xr) \markov (\yd,\yr), \quad \trace(\Cs) \leq \PS, \quad \trace(\Cr) \leq \PR
\end{multline}
is obtained as shown in~\cite[Section~9.4.1]{Kramer07TopicsinMulti-User}. Here, $\u$ is an auxiliary variable representing the part of the source message the relay must decode. In addition to the power constraints, the maximization over the joint distribution of $\u$, $\xs$, and $\xr$ is subject to the constraint that $\u \markov (\xs,\xr) \markov (\yr,\yd)$ forms a Markov chain.

It is easy to see that DF is a special case of PDF for which $\u = \xs$. In particular, choosing $\u = \xs$ in~\eqref{eq:pdf-opt} yields that
\begin{align}\label{eq:df}
	\rdf = \max_{p(\xs,\xr)} \min\, \bigl\{ 
	I(\xs; \yr | \xr), I(\xs, \xr; \yd) \bigr\}
	\quad \st \quad \trace(\Cs) \leq \PS, \quad \trace(\Cr) \leq \PR
\end{align}
is achievable by means of the PDF strategy, where $\rdf$ denotes the maximum achievable DF rate for the Gaussian MIMO relay channel, cf.~\cite[Theorem~1]{Cover79CapacityTheoremsRelay}. Moreover, choosing $\u = \zeros$ yields the point-to-point~(P2P) capacity of the source-to-destination link
\begin{equation}\begin{split}
	\rp 
		&= \max_{p(\xs)} \, I(\xs; \yd | \xr=\zeros)
		\quad \st \quad \trace(\Cs) \leq \PS.
\end{split}\end{equation}
Therefore, it is clear that PDF always achieves at least the maximum of the rates that are achievable by means of DF and by means of direct transmission from source to destination, i.e.,
\begin{equation}
	\rpdf \geq \max \left\{ \rdf, \rp \right\}.
\end{equation}

In contrast to $\rdf$, $\rp$, and the cut-set bound~(CSB), which is given by~\cite[Theorem~4]{Cover79CapacityTheoremsRelay}
\begin{align}\label{eq:csb}
	\! \csb = \max_{p(\xs,\xr)} \min\, \bigl\{ 
	I(\xs; \yr, \yd | \xr), I(\xs, \xr; \yd) \bigr\}
	\quad \st \quad \trace(\Cs) \leq \PS, \quad \trace(\Cr) \leq \PR,
\end{align}
we cannot simply invoke the entropy maximizing property of the zero-mean proper (circularly symmetric) complex Gaussian distribution (cf.~\cite{Telatar99CapacityofMulti-antenna, Neeser93ProperComplexRandom}) to argue that $\rpdf$ is maximized by  Gaussian channel inputs. The reason for this is that the entropy maximizing property cannot be applied to the term 
\begin{align}\begin{split}
	I(\u; \yr | \xr) + I(\xs; \yd | \u, \xr) &= 
	h(\Hsr \xs + \nr | \xr) - h(\nd) \\
	&+ h(\Hsd \xs + \nd | \u, \xr) - h(\Hsr \xs + \nr | \u, \xr),
\end{split}\end{align}
which includes the difference $h(\Hsd \xs + \nd | \u, \xr) - h(\Hsr \xs + \nr | \u, \xr)$ of two conditional differential entropies involving $\u$, $\xs$, and $\xr$.

For $\Hsr = \Hsd = \id$, the maximization of such a difference over the conditional probability distribution $p(\xs | \u, \xr)$ subject to shaping constraints on the conditional covariance matrix $\E[\xs \xs^{\he} | \u, \xr]$ was analyzed in~\cite[Theorem~8]{Liu07ExtremalInequalityMotivated}, where it is proved that the optimal distribution is Gaussian. However, the term $I(\u; \yr | \xr) + I(\xs; \yd | \u, \xr)$, and hence the difference $h(\Hsd \xs + \nd | \u, \xr) - h(\Hsr \xs + \nr | \u, \xr)$, is only one part of the objective function of the PDF rate maximization problem given in~\eqref{eq:pdf-opt}. Therefore, we cannot directly apply~\cite[Theorem~8]{Liu07ExtremalInequalityMotivated} to prove that the achievable PDF rate is maximized by Gaussian channel inputs. Rather, we need to establish achievability and converse for the whole objective function.

\section{Aligned Gaussian MIMO Relay Channel}

As a first step towards a characterization of the maximum achievable PDF rate for the Gaussian MIMO relay channel, we consider the \emph{aligned} Gaussian MIMO relay channel. The results for this special case are then generalized in the following section.

\begin{definition}
The Gaussian MIMO relay channel is said to be aligned if $\Ns = \Nr = \Nd = N$ and $\Hsr = \Hsd = \id_{N}$.
\end{definition}

The channel model for the aligned Gaussian MIMO relay channel is hence given by
\begin{align}\begin{alignedat}{3}\label{eq:alignedRC}
	\yr &= \xs + \nr, \qquad& \nr &\sim \CN(\zeros,\Zr), \\
	\yd &= \xs + \Hrd \xr + \nd, \qquad& \nd &\sim \CN(\zeros,\Zd).
\end{alignedat}\end{align}
As the theorem below reveals, Gaussian channel inputs maximize the achievable PDF rate for this particular relay channel. 

\begin{theorem}\label{thm:aligned}
For the aligned Gaussian MIMO relay channel, the maximum achievable PDF rate is attained by jointly proper complex Gaussian source and relay inputs.
\end{theorem}
\begin{proof}
\red{\emph{Achievability:}}
Let $\q \sim \CN(\zeros,\Cq)$, $\v \sim \CN(\zeros,\Cv)$, $\xr \sim \CN(\zeros,\Cr)$ be independent, $\u = \q + \A\xr$, and $\xs = \u + \v$ such that $\xs \sim \CN(\zeros,\Cs)$ with $\Cs = \Cq + \A \Cr \A^{\he} + \Cv$. Then,
\begin{multline}
\begin{aligned}
	\rpdf \geq \rpdfn = \max_{\Cq, \Cv, \Cr, \A} \, \min \, \Biggl\{
	&\log \frac{|\Cq + \Cv + \Zr|}{|\Cv + \Zr|} + \log \frac{|\Cv + \Zd|}{|\Zd|}, \\
	&\log \frac{|\Cq + \Cv + (\Hrd+\A) \Cr (\Hrd+\A)^{\he} + \Zd|}{|\Zd|}	 \Biggr\}
\end{aligned}\\
	\st \quad \Cq, \Cv, \Cr \psd \zeros, \quad \trace(\Cq + \Cv + \A \Cr \A^{\he}) \leq \PS, \quad \trace(\Cr) \leq \PR,
\end{multline}
where $\rpdfn$ denotes a PDF rate that is achievable with proper complex Gaussian channel inputs. \red{By introducing an auxiliary variable $\S = \Cq + \Cv \psd \zeros$, this achievable rate can equivalently be expressed as
\begin{multline}
\begin{aligned}
	\rpdfn = \max_{\S, \Cq, \Cv, \Cr, \A} \, \min \, \Biggl\{
	&\log \frac{|\S + \Zr|}{|\Cv + \Zr|} + \log \frac{|\Cv + \Zd|}{|\Zd|}, \\
	&\log \frac{|\S + (\Hrd+\A) \Cr (\Hrd+\A)^{\he} + \Zd|}{|\Zd|}	 \Biggr\}
\end{aligned}\\
	\st \quad \S, \Cq, \Cv, \Cr \psd \zeros, \quad \Cq + \Cv = \S, \quad 
	\trace(\S + \A \Cr \A^{\he}) \leq \PS, \quad \trace(\Cr) \leq \PR.
\end{multline}
If we now apply the primal decomposition approach that was already considered in~\cite{Hellings14OptimalGaussianSignaling} to this problem, we obtain
\begin{align}\label{eq:primal-decomposition}
	\rpdfn = \max_{\S} \, \rpdfn(\S) \quad \st \quad \S \psd \zeros, 
	\quad \trace(\S) \leq \PS
\end{align}
with
\begin{multline}
\begin{aligned}
	\rpdfn(\S) = \max_{\Cq, \Cv, \Cr, \A} \, \min \, \Biggl\{
	&\log \frac{|\S + \Zr|}{|\Zd|} + \log \frac{|\Cv + \Zd|}{|\Cv + \Zr|}, \\
	&\log \frac{|\S + (\Hrd+\A) \Cr (\Hrd+\A)^{\he} + \Zd|}{|\Zd|} \Biggr\}
\end{aligned}\\
	\st \quad \Cq, \Cv, \Cr \psd \zeros, \quad \Cq + \Cv = \S, \quad 
	\trace(\A \Cr \A^{\he}) \leq \PS - \trace(\S), \quad \trace(\Cr) \leq \PR.
\end{multline}
Note that $\Cq$ only appears in the constraints $\Cq \psd \zeros$ and $\Cq + \Cv = \S$, i.e., it is a slack variable and can be eliminated, after which the equality constraint becomes $\Cv \nsd \S$. Furthermore, $\Cv$ only contributes to the second summand of the first term inside the minimum of the objective function so that $\rpdfn(\S)$ is equal to
\begin{multline}
\begin{aligned}
	\rpdfn(\S) = \max_{\Cr, \A} \, \min \, \Biggl\{
	&\log \frac{|\S + \Zr|}{|\Zd|} + \max_{\zeros \nsd \Cv \nsd \S} \ 
	\log \frac{|\Cv + \Zd|}{|\Cv + \Zr|}, \\
	&\log \frac{|\S + (\Hrd+\A) \Cr (\Hrd+\A)^{\he} + \Zd|}{|\Zd|} \Biggr\}
\end{aligned}\\
	\st \quad \Cr \psd \zeros, \quad \trace(\A \Cr \A^{\he}) \leq \PS - \trace(\S),
	\quad \trace(\Cr) \leq \PR.
\end{multline}
}

In order to \red{further} simplify this expression, consider the inner maximization problem
\begin{align}\label{eq:problemCv}
	\max_{\Cv} \, \log \frac{|\Cv + \Zd|}{|\Cv + \Zr|} \quad \st \quad 
	\zeros \nsd \Cv \nsd \S,
\end{align}
which up to the additive constant $\log \left(|\Zr|/|\Zd|\right)$ is mathematically equivalent to the problem that yields the secrecy capacity of the aligned Gaussian MIMO wiretap channel (vector Gaussian wiretap channel) under shaping constraints, cf.~\cite[Section~II-A]{Liu09NoteSecrecyCapacity}.\footnote{Comparing~\eqref{eq:problemCv} to~\cite[eq.~(17)]{Liu09NoteSecrecyCapacity}, we see that the destination plays the role of the legitimate receiver and the relay that of the eavesdropper.} Following the proof of~\cite[Theorem~2]{Liu09NoteSecrecyCapacity}, which carries over to the complex-valued setting under consideration here, we can determine the optimal value of the inner problem~\eqref{eq:problemCv}.

To this end, first note that the Karush--Kuhn--Tucker (KKT) conditions are necessary for problem~\eqref{eq:problemCv} since the Abadie constraint qualification is automatically satisfied if all constraints are linear~\cite[Section~5.1]{Bazaraa2006} and since the KKT conditions readily extend to problems with generalized inequalities such as positive semidefiniteness constraints~\cite[Section~5.9.2]{Boyd2004}. Thus, any optimizer $\Cv^{\star}$ of problem~\eqref{eq:problemCv} must satisfy
\begin{align}
	(\Cv^{\star} + \Zd)^{-1} + \lam_{1} &= (\Cv^{\star} + \Zr)^{-1} + \lam_{2},
	 \label{eq:DF} \\
	\Cv^{\star} \lam_{1} &= \zeros, \label{eq:CS1} \\
	(\S - \Cv^{\star}) \lam_{2} &= \zeros,  \label{eq:CS2}
\end{align}
where $\lam_{1} \psd \zeros$ and $\lam_{2} \psd \zeros$ denote the Lagrangian multipliers corresponding to the (generalized) inequality constraints $\Cv \psd \zeros$ and $\S - \Cv \psd \zeros$, respectively. Now, let $\Z$ such that 
\begin{align}\label{eq:defZ}
	(\Cv^{\star} + \Z)^{-1} = (\Cv^{\star} + \Zd)^{-1} + \lam_{1}.
\end{align}
It then follows from~\eqref{eq:CS1} that an explicit expression for $\Z$ (as a function of $\Zd$ and the Lagrangian multiplier $\lam_{1}$) is given by
\begin{align}\label{eq:Z}
	\Z = (\Zd^{-1} + \lam_{1})^{-1}.
\end{align}
Since $\lam_{1} \psd \zeros$, we can conclude that $\Z \posd \zeros$. Furthermore,~\eqref{eq:DF} and the definition of $\Z$ in~\eqref{eq:defZ} imply that
\begin{align}\label{eq:defZ2}
	(\Cv^{\star} + \Z)^{-1} = (\Cv^{\star} + \Zr)^{-1} + \lam_{2}.
\end{align}

By means of the variable $\Z$, we can characterize the optimal value of problem~\eqref{eq:problemCv} as follows. First, note that
\begin{align}\begin{split}
	(\Cv^{\star} + \Z) \Z^{-1}
	&= \Cv^{\star} \Z^{-1} + \id \\
	&\overset{\eqref{eq:Z}}{=} \Cv^{\star} (\Zd^{-1} + \lam_{1}) + \id \\
	&\overset{\eqref{eq:CS1}}{=} \Cv^{\star} \Zd^{-1} + \id \\
	&= (\Cv^{\star} + \Zd) \Zd^{-1},
\end{split}\end{align}
which implies
\begin{align}\label{eq:equiv2}
	\frac{|\Cv^{\star} + \Zd|}{|\Zd|} = \frac{|\Cv^{\star} + \Z|}{|\Z|}.
\end{align}
Similarly, it holds that
\begin{align}\begin{split}
	(\S + \Z) (\Cv^{\star} + \Z)^{-1}
	&= (\S - \Cv^{\star} + \Cv^{\star} + \Z) (\Cv^{\star} + \Z)^{-1} \\
	&= (\S - \Cv^{\star}) (\Cv^{\star} + \Z)^{-1} + \id \\
	&\overset{\eqref{eq:defZ2}}{=} (\S - \Cv^{\star}) \left( (\Cv^{\star} + \Zr)^{-1} + \lam_{2} \right) + \id \\
	&\overset{\eqref{eq:CS2}}{=} (\S - \Cv^{\star}) (\Cv^{\star} + \Zr)^{-1} + \id \\
	&= (\S - \Cv^{\star} + \Cv^{\star} + \Zr) (\Cv^{\star} + \Zr)^{-1} \\
	&= (\S + \Zr) (\Cv^{\star} + \Zr)^{-1},
\end{split}\end{align}
from which we obtain
\begin{align}\label{eq:equiv1}
	\frac{|\Cv^{\star} + \Z|}{|\Cv^{\star} + \Zr|} = 
	\frac{|\S + \Z|}{|\S + \Zr|}.
\end{align}
The optimal value of problem~\eqref{eq:problemCv} can be therefore calculated as
\begin{align}\begin{split}
	\log \frac{|\Cv^{\star} + \Zd|}{|\Cv^{\star} + \Zr|}
	&= \log \frac{|\Cv^{\star} + \Zd|}{|\Zd|} 
		- \log \frac{|\Cv^{\star} + \Zr|}{|\Zd|} \\
	&\overset{\eqref{eq:equiv2}}{=} \log \frac{|\Cv^{\star} + \Z|}{|\Z|} 
		- \log \frac{|\Cv^{\star} + \Zr|}{|\Zd|} \\
	&= \log \frac{|\Cv^{\star} + \Z|}{|\Cv^{\star} + \Zr|} 
		- \log \frac{|\Z|}{|\Zd|} \\
	&\overset{\eqref{eq:equiv1}}{=} \log \frac{|\S + \Z|}{|\S + \Zr|} 
		- \log \frac{|\Z|}{|\Zd|} \\
	&= \log \frac{|\S + \Z|}{|\Z|}
		- \log \frac{|\S + \Zr|}{|\Zd|}.
\end{split}\end{align}
Using this result, it is straightforward to verify that $\rpdfn(\S)$ is equal to
\begin{multline}\label{eq:achievable-OP}
	\rpdfn(\S) = \max_{\Cr, \A} \, \min \, \Biggl\{
	\log \frac{|\S + \Z|}{|\Z|}, \log \frac{|\S + (\Hrd+\A) \Cr (\Hrd+\A)^{\he} + \Zd|}{|\Zd|}	 \Biggr\} \\
	\st \quad \Cr \psd \zeros, \quad \trace(\A \Cr \A^{\he}) \leq \PS - \trace(\S), \quad \trace(\Cr) \leq \PR
\end{multline}
with $\Z$ from~\eqref{eq:Z}.

\textit{Converse}: The converse of the proof is based on the so-called \emph{channel enhancement} technique, which was originally introduced in~\cite{Weingarten06CapacityRegionof}. From~\eqref{eq:defZ},~\eqref{eq:defZ2}, and the positive semidefiniteness of the Lagrangian multipliers $\lam_{1}, \lam_{2}$, it follows that
\begin{align}
	\Z \nsd \Zd, \qquad \Z \nsd \Zr.
\end{align}
Consequently, we can use $\Z$ to define an enhanced aligned Gaussian MIMO relay channel. In particular, let $\Zrt = \Z$ and
\begin{align}\begin{alignedat}{3}\label{eq:enhanced-alignedRC}
	\ntilde{\yr} &= \xs + \ntilde{\nr}, \qquad& 
	\ntilde{\nr} &\sim \CN(\zeros,\Zrt), \\
	\yd &= \xs + \Hrd \xr + \nd, \qquad& 
	\nd &\sim \CN(\zeros,\Zd).
\end{alignedat}\end{align}
Since $\Z \nsd \Zr$, $\yr$ is a stochastically degraded version of $\ntilde{\yr}$ so that $I(\u; \yr | \xr) \leq I(\u; \ntilde{\yr} | \xr)$ \red{for all feasible $p(\u,\xs,\xr)$}. Therefore,
\begin{multline}\label{eq:converseI}
	\rpdf \leq \ntilde{\rpdf} = \max_{p(\u,\xs,\xr)} \min\, \bigl\{ I(\u; \ntilde{\yr} | \xr) + I(\xs; \yd | \u,\xr), I(\xs, \xr; \yd) \bigr\} \\
	\st \quad \u \markov (\xs,\xr) \markov (\yd,\ntilde{\yr}), \quad \trace(\Cs) \leq \PS, \quad \trace(\Cr) \leq \PR,
\end{multline}
which explains why we call the relay channel defined in~\eqref{eq:enhanced-alignedRC} enhanced.

\red{Moreover, given $\xr$, $\yd$ is a stochastically degraded version of $\ntilde{\yr}$ as well. In fact, since $\Z \nsd \Zd$, the enhanced aligned Gaussian MIMO relay channel belongs to the class of stochastically degraded relay channels according to the definition in~\cite{Gerdes14OptimalPartialDecode-and-Forward}.} From~\cite[Proposition~1]{Gerdes14OptimalPartialDecode-and-Forward}, it hence follows that the optimal PDF strategy for the enhanced relay channel~\eqref{eq:enhanced-alignedRC} is equivalent to DF, i.e., $\ntilde{\rpdf} = \ntilde{\rdf}$ with
\begin{align}
	\ntilde{\rdf} = \max_{p(\xs,\xr)} \min\, \bigl\{ 
	I(\xs; \ntilde{\yr} | \xr), I(\xs, \xr; \yd) \bigr\} 
	\quad \st \quad \trace(\Cs) \leq \PS, \quad \trace(\Cr) \leq \PR.
\end{align}
However, the maximum achievable DF rate for the Gaussian MIMO relay channel is attained by jointly proper complex Gaussian channel inputs~\cite{Wang05CapacityofMIMO}, which essentially follows from the fact that the zero-mean proper (circularly symmetric) complex Gaussian distribution maximizes the differential entropy, cf.~\cite{Neeser93ProperComplexRandom, Telatar99CapacityofMulti-antenna}.

Therefore, the achievable PDF rate for the enhanced aligned Gaussian MIMO relay channel is maximized by letting $\q \sim \CN(\zeros,\Cq)$ and $\xr \sim \CN(\zeros,\Cr)$ be independent and $\xs = \q + \A \xr$ such that $\xs \sim \CN(\zeros,\Cs)$ with $\Cs = \Cq + \A \Cr \A^{\he}$, i.e.,
\begin{multline}
	\ntilde{\rpdf} = \max_{\Cq, \Cr, \A} \, \min \, \Biggl\{
	\log \frac{|\Cq + \Z|}{|\Z|}, \log \frac{|\Cq  + (\Hrd+\A) \Cr (\Hrd+\A)^{\he} + \Zd|}{|\Zd|}	 \Biggr\} \\
	\st \quad \Cq, \Cr \psd \zeros, \quad \trace(\Cq + \A \Cr \A^{\he}) \leq \PS, \quad \trace(\Cr) \leq \PR.
\end{multline}
Using a primal decomposition again, this maximization problem can equivalently be written as
\begin{align}\label{eq:primal-decomposition2}
	\ntilde{\rpdf} = \max_{\Cq} \, \ntilde{\rpdf(\Cq)} \quad \st \quad \Cq \psd \zeros, 
	\quad \trace(\Cq) \leq \PS,
\end{align}
where
\begin{multline}\label{eq:enhanced-aligned-OP}
	\ntilde{\rpdf(\Cq)} = \max_{\Cr, \A} \, \min \, \Biggl\{
	\log \frac{|\Cq + \Z|}{|\Z|}, \log \frac{|\Cq + (\Hrd+\A) \Cr (\Hrd+\A)^{\he} + \Zd|}{|\Zd|} \Biggr\} \\
	\st \quad \Cr \psd \zeros, \quad \trace(\A \Cr \A^{\he}) \leq \PS - \trace(\Cq), \quad \trace(\Cr) \leq \PR.
\end{multline}
Comparing~\eqref{eq:achievable-OP} to~\eqref{eq:enhanced-aligned-OP}, we notice that $\ntilde{\rpdf(\Cq)} = \rpdfn(\S)$ for $\Cq = \S$, from which we can directly conclude that $\ntilde{\rpdf} = \rpdfn$ as the constraints in~\eqref{eq:primal-decomposition} and~\eqref{eq:primal-decomposition2} are the same. But since $\ntilde{\rpdf} \geq \rpdf \geq \rpdfn$ in general, it follows that $\rpdf = \rpdfn$.
\end{proof}

\section{General Gaussian MIMO Relay Channel}

In this section, we extend the result that jointly proper complex Gaussian source and relay inputs maximize the achievable PDF rate to the general Gaussian MIMO relay channel. The main idea for extending the proof from the aligned to the general case is adopted from~\cite{Weingarten06CapacityRegionof}. First, we write the channel model of the Gaussian MIMO relay channel~\eqref{eq:RC} in an equivalent form with square channel gain matrices. In a second step, we use the singular value decomposition~(SVD) to enhance $\Hsr$ and $\Hsd$ by adding small perturbations to their singular values such that the resulting channel gain matrices are invertible. Finally, we show that the maximum achievable PDF rate for the original Gaussian MIMO relay channel can be obtained by a limit process on the maximum achievable PDF rate for the enhanced (perturbed) relay channel.

\begin{theorem}\label{thm:general}
For the Gaussian MIMO relay channel, the maximum achievable PDF rate is attained by jointly proper complex Gaussian source and relay inputs.
\end{theorem}
\begin{proof}
Without loss of generality, we may assume that $\Hsr, \Hsd, \Hrd \in \setC^{N \times N}$ with $N = \max \left\{ \Ns, \Nr, \Nd \right\}$. If this were not the case, we could augment the matrices with zeros to obtain square $N \times N$ channel gain matrices without changing the achievable PDF rate. Furthermore, we may also assume that $\Zr = \Zd = \id_{N}$ since any Gaussian MIMO relay channel with nonsingular noise covariances can be transformed into one with additive white Gaussian noise by means of a noise whitening operation, cf.~\cite{Goldsmith03CapacityLimitsof}.

\red{\emph{Achievability:}}
Let $\q \sim \CN(\zeros,\Cq)$, $\v \sim \CN(\zeros,\Cv)$, $\xr \sim \CN(\zeros,\Cr)$ be independent, $\u = \q + \A\xr$, and $\xs = \u + \v$ such that $\xs \sim \CN(\zeros,\Cs)$ with $\Cs = \Cq + \A \Cr \A^{\he} + \Cv$. Then, the PDF rate
\begin{multline}\label{eq:rpdf-achievable}
	\rpdfn = \max_{\Cq, \Cv, \Cr, \A} \, \rpdfn(\Cq, \Cv, \Cr, \A) \\
	\st \quad \Cq, \Cv, \Cr \psd \zeros, \quad \trace(\Cq + \Cv + \A \Cr \A^{\he}) \leq \PS, \quad \trace(\Cr) \leq \PR
\end{multline}
is achievable with jointly proper complex Gaussian source and relay inputs, where
\begin{align}\begin{split}
	\rpdfn(\Cq, \Cv, \Cr, \A) = \min \, \Biggl\{
	&\log \frac{|\id + \Hsr(\Cq + \Cv)\Hsr^{\he}|}
	{|\id + \Hsr \Cv \Hsr^{\he}|} 
	+ \log |\id + \Hsd \Cv \Hsd^{\he}|, \\
	&\log |\id + \Hsd(\Cq + \Cv)\Hsd^{\he} 
	+ \Ha(\A) \Cr \Ha(\A)^{\he}| \Biggr\}
\end{split}\end{align}
and $\Ha(\A) = \Hrd+\Hsd\A$. 

\textit{Converse}: Suppose the SVDs of $\Hsr$ and $\Hsd$ are given by
\begin{align}
	\Hsr = \Usr \Ssr \Vsr^{\he}, \qquad \Hsd = \Usd \Ssd \Vsd^{\he},
\end{align}
where $\Usr, \Usd, \Vsr, \Vsd \in \setC^{N \times N}$ are unitary and the diagonal matrices $\Ssr, \Ssd \in \setR_{+}^{N \times N}$ contain the singular values of $\Hsr, \Hsd$. For some $\varepsilon > 0$, let
\begin{align}
	\nbar{\Hsr} = \Usr (\Ssr + \varepsilon \id) \Vsr^{\he}, \qquad 
	\nbar{\Hsd} = \Usd (\Ssd + \varepsilon \id) \Vsd^{\he},
\end{align}
and consider the following enhanced Gaussian MIMO relay channel:
\begin{align}\begin{alignedat}{3}\label{eq:enhancedRC}
	\nbar{\yr} &= \nbar{\Hsr} \xs + \nr, \qquad& \nr &\sim \CN(\zeros,\id_{N}), \\
	\nbar{\yd} &= \nbar{\Hsd} \xs + \Hrd \xr + \nd, \qquad& \nd &\sim \CN(\zeros,\id_{N}).
\end{alignedat}\end{align}
As both $\nbar{\Hsr}$ and $\nbar{\Hsd}$ are invertible, this relay channel is equivalent to an aligned Gaussian MIMO relay channel with
\begin{align}\label{eq:equiv-noise}
	\Zr = (\nbar{\Hsr}^{\he} \nbar{\Hsr})^{-1}, \qquad
	\Zd = (\nbar{\Hsd}^{\he} \nbar{\Hsd})^{-1},
\end{align}
for which we know from Theorem~\ref{thm:aligned} that proper complex Gaussian channel inputs maximize the achievable PDF rate.

Furthermore, note that
\begin{align}\begin{alignedat}{3}\label{eq:enhanced-channels}
	\Hsr &= \Bsr \nbar{\Hsr}, \qquad&
	\Bsr &= \Usr \Ssr (\Ssr + \varepsilon \id)^{-1} \Usr^{\he} \nsd \id_{N}, \\
	\Hsd &= \Bsd \nbar{\Hsd}, \qquad&
	\Bsd &= \Usd \Ssd (\Ssd + \varepsilon \id)^{-1} \Usd^{\he} \nsd \id_{N}.
\end{alignedat}\end{align}
Since~\eqref{eq:enhanced-channels} is equivalent to
\begin{align}
	\Hsr^{\he} \Hsr \nsd \nbar{\Hsr}^{\he} \nbar{\Hsr}, \qquad
	\Hsd^{\he} \Hsd \nsd \nbar{\Hsd}^{\he} \nbar{\Hsd}
\end{align}
due to~\cite[Lemma~5]{Shang13Noisy-InterferenceSum-RateCapacity}, this means that $\yr$ and $\yd$ are stochastically degraded versions of $\nbar{\yr}$ and $\nbar{\yd}$, respectively, cf.~\eqref{eq:equiv-noise}. As a result, we have
\begin{align}
	\nbar{\rpdfn} = \nbar{\rpdf} \geq \rpdf \geq \rpdfn,
\end{align}
where
\begin{multline}\label{eq:enhanced-PDF}
	\nbar{\rpdfn} = \max_{\Cq, \Cv, \Cr, \A} \,
	\nbar{\rpdfn}(\Cq, \Cv, \Cr, \A) \\
	\st \quad \Cq, \Cv, \Cr \psd \zeros, \quad \trace(\Cq + \Cv + \A \Cr \A^{\he}) \leq \PS, \quad \trace(\Cr) \leq \PR
\end{multline}
is the maximum achievable PDF rate for the enhanced Gaussian MIMO relay channel defined in~\eqref{eq:enhancedRC},
\begin{align}\begin{split}
	\nbar{\rpdfn}(\Cq, \Cv, \Cr, \A) = \min \, \Biggl\{
	&\log \frac{|\id + \nbar{\Hsr}(\Cq + \Cv)\nbar{\Hsr}^{\he}|}
	{|\id + \nbar{\Hsr} \Cv \nbar{\Hsr}^{\he}|} 
	+ \log |\id + \nbar{\Hsd} \Cv \nbar{\Hsd}^{\he}|, \\
	&\log |\id + \nbar{\Hsd}(\Cq + \Cv)\nbar{\Hsd}^{\he} 
	+ \bar{\Ha}(\A) \Cr \bar{\Ha}(\A)^{\he}| \Biggr\},
\end{split}\end{align}
and $\bar{\Ha}(\A) = \Hrd + \nbar{\Hsd} \A$. The proof can hence be completed by showing that $\nbar{\rpdfn} \to \rpdfn$ as $\varepsilon \to 0$.

To this end, suppose $\Cq, \Cv, \Cr, \A$ are fixed. Then, $\nbar{\rpdfn}(\Cq, \Cv, \Cr, \A)$ is a continuous function of $\varepsilon$ because it is the pointwise minimum of two functions that are continuous in $\varepsilon$. As a consequence,
\begin{align}\label{eq:epsilon-limit}
	\lim_{\varepsilon \to 0} \, \nbar{\rpdfn}(\Cq, \Cv, \Cr, \A) 
	= \rpdfn(\Cq, \Cv, \Cr, \A).
\end{align}
Because~\eqref{eq:epsilon-limit} holds for any $\A$ and any positive semidefinite $\Cq, \Cv, \Cr$, it also holds for the maximizers $\Cq^{\star}, \Cv^{\star}, \Cr^{\star}, \A^{\star}$ of problem~\eqref{eq:enhanced-PDF}. In addition, these maximizers also satisfy the constraints of~\eqref{eq:rpdf-achievable}, which means that
\begin{align}
	\lim_{\varepsilon \to 0} \, \nbar{\rpdfn} = 
	\lim_{\varepsilon \to 0} \,\nbar{\rpdfn}(\Cq^{\star}, \Cv^{\star}, \Cr^{\star}, \A^{\star}) =
	\rpdfn(\Cq^{\star}, \Cv^{\star}, \Cr^{\star}, \A^{\star}) \leq \rpdfn.
\end{align}
But since $\nbar{\rpdfn} \geq \rpdfn$ in general, this implies $\lim_{\varepsilon \to 0} \, \nbar{\rpdfn} = \rpdfn$.
\end{proof}

\section{Discussion}\label{sec:discussion}

In this paper, we showed that the maximum achievable PDF rate for the Gaussian MIMO relay channel is always attained by jointly proper complex Gaussian source and relay inputs. The main challenge of proving this result (Theorem~\ref{thm:general}) was to establish that Gaussian channel inputs maximize the achievable PDF rate for the aligned Gaussian MIMO relay channel (Theorem~\ref{thm:aligned}). The general result then followed from a rather simple limiting argument.

\subsection{Comments on the Proof of Theorem~\ref{thm:aligned}}

The key to proving Theorem~\ref{thm:aligned} was to employ a primal decomposition approach in the achievability part. This is because using the primal decomposition approach leads to an inner maximization problem which is mathematically equivalent to the optimization problem that yields the secrecy capacity of the aligned Gaussian MIMO wiretap channel (vector Gaussian wiretap channel) under shaping constraints~\cite[Section~II-A]{Liu09NoteSecrecyCapacity}. In particular, we were able to obtain both the optimal value of this inner problem, which we used to simplify the PDF rate maximization problem in the achievability part, and the enhanced channel, which we required for the converse, from considerations similar to those in the proof of~\cite[Theorem~2]{Liu09NoteSecrecyCapacity}.

However, note that there is an important difference in how the channel enhancement argument was used in the converse parts of~Theorem~\ref{thm:aligned} and~\cite[Theorem~2]{Liu09NoteSecrecyCapacity}. Whereas the secrecy capacity of the aligned Gaussian MIMO wiretap channel was derived by enhancing the channel to the legitimate receiver, we enhanced the channel from the source to the relay, which plays the role of the eavesdropper in the inner problem~\eqref{eq:problemCv}. The reason for this is that enhancing the source-to-destination channel would not have yielded the desired converse for our purpose as the second mutual information term inside the minimum of~\eqref{eq:converseI} would also have increased. By enhancing the source-to-relay channel, this term remained unchanged. Moreover, we obtained a stochastically degraded Gaussian MIMO relay channel, for which it is known that the optimal PDF strategy is equivalent to DF~\cite[Proposition~1]{Gerdes14OptimalPartialDecode-and-Forward}.

Like for the aligned Gaussian MIMO wiretap channel, the existence and the properties of the enhanced channel can be explained by considering the special case of parallel subchannels, i.e., the special case where the noise covariances are diagonal, cf.~\cite[Section~III]{Liu09NoteSecrecyCapacity}. In particular, the achievable PDF rate for the parallel Gaussian relay channel is maximized if the relay decodes the entire information transmitted over the subchannels for which the relay receives a better signal than the destination, but no information sent over the subchannels where the relay's receive signal is worse~\cite{Liang07ResourceAllocationWireless}. That is, on subchannels where the source-to-relay channel is better than the source-to-destination channel, the optimal PDF strategy reduces to DF, whereas it is equivalent to P2P transmission if the source-to-destination channel is better than the source-to-relay channel. Therefore, an enhanced and stochastically degraded relay channel can be constructed as follows: For any subchannel where the destination receives a better signal than the relay, the noise variance of the relay is reduced to that of the destination. The resulting enhanced parallel Gaussian channel is stochastically degraded as, on every subchannel, the destination's receive signal is no better than the relay's receive signal. Furthermore, the maximum achievable PDF rate does not increase compared to the original parallel Gaussian relay channel, but it can now also be achieved by letting the relay decode the entire information transmitted by the source.

Following this line of thought, one can think of the enhanced aligned Gaussian MIMO relay channel~\eqref{eq:enhanced-alignedRC} as the relay channel that is obtained by reducing the noise covariance of the relay ``just enough'' in the sense that $\rpdf$ does not increase compared to the original aligned relay channel~\eqref{eq:alignedRC} and that it can be achieved by means of a pure DF strategy. Because the noise covariances $\Zr$ and $\Zd$ may have different eigendirections, finding the enhanced channel is more involved than for the parallel Gaussian relay channel. However, the fact that such an enhanced channel always exists can be concluded from the arguments we adopted from the work on the secrecy capacity of the aligned Gaussian MIMO wiretap channel, cf.~\cite[Theorem~2]{Liu09NoteSecrecyCapacity}.

\subsection{Interpretation of the Primal Decomposition Approach}

The matrix $\S \psd \zeros$, which was introduced as part of the primal decomposition approach, has a nice interpretation, cf.~\cite{Hellings14OptimalGaussianSignaling}. In fact, note that $\rpdf$ is achieved by a block Markov superposition encoding scheme (with one block memory) that uses $B$ blocks of transmission to convey $B-1$ independent messages from the source to the destination. In block $b$, the source splits its message $w_{b}$ into independent parts $w_{b}'$ and $w_{b}''$. The first part $w_{b}'$ is decoded by the relay, whereas the second part $w_{b}''$ is intended for the destination only. Assuming the relay encoder operates in a causal manner, the relay’s transmit signal $\xr$ in block $b$ is a function of the previous message $w_{b-1}'$, and provided that the relay encoding function is deterministic, $\xr$ is then also known to the source. If we let $\q$ and $\v$ be functions of the current message parts $w_{b}'$ and $w_{b}''$, respectively, the transmit signal $\xs(w_{b}', w_{b}'',w_{b-1}')$ of the source in block $b$ can be expressed as the superposition of the independent signal parts $\q(w_{b}')$, $\v(w_{b}'')$, and $\xr(w_{b-1}')$ given by
\begin{align}\label{eq:decomposition-xs}
	\xs = \q + \A \xr + \v.
\end{align}

\begin{figure}[t]
	\centering
	\begin{tikzpicture}[scale=0.5,>=latex]
		\tikzset{device/.style={%
			ellipse, line width=1.25pt, minimum width=3em, minimum height=3em, draw}}                     
		\tikzset{signal/.style=%
			{rectangle, inner sep=.1em, draw=none}}                     
		\tikzset{SP/.style=%
			{rectangle,line width=1pt,minimum width=.5em,minimum height=.5em, draw}}
		\node[device,minimum height=8.2em,minimum width=4em] (source) at (0,2) {};
		\draw (source)++(-1,0) node[] {S};
		\node[signal] (sourceAxr) at (0,4) {$\A\xr$};
		\node[SP] (sourceSP) at (0,2.5) {};
		\node[signal] (sourceq) at (0,1) {$\q$};
		\node[signal] (sourcev) at (0,0) {$\v$};
		\node[device,minimum width=4.95em] (relay) at (7.25,4) {};
		\draw (relay)++(0,0.65) node[] {R};
		\draw (relay)++(-.75,0) node[SP] (relaySP) {};
		\draw (relay)++(0.75,0) node[signal] (relayxr) {$\xr$};
		\node[device] (destination) at (8,0) {D};
		\draw[line width=1pt,->] (sourceq) to (sourceSP);
		\draw[line width=1pt,->] (sourceSP) to (sourceAxr);
		\draw[black!30!white,line width=1pt,->] (sourceAxr) to (relay);
		\draw[line width=1pt,->] (sourceq) to (relaySP.west);
		\draw[dashed,line width=1pt,->] (sourcev) to (relay);
		\draw[line width=1pt,->] (sourceAxr) to (7.5,2.25) to (destination);
		\draw[line width=1pt,->] (sourceq) to (destination);
		\draw[line width=1pt,->] (sourcev) to (destination);
		\draw[line width=1pt,->] (relaySP.east) to (relayxr);
		\draw[line width=1pt,->] (relayxr) to (destination);
		\node[signal] (coop) at (6.5,1.75) {\small coop.};
		\ifCLASSOPTIONdraftcls
			\draw[line width=1pt,->] (11,3.5) to ++(1.2,0) node[right]
				{\small useful signal};
			\draw[black!30!white,line width=1pt,->] (11,2.5) to ++(1.2,0) node[right]
				{\small \textcolor{black}{known interference\phantom{g}}};
			\draw[dashed,line width=1pt,->] (11,1.5) to ++(1.2,0) node[right]
				{\small interference\phantom{g}};
			\draw (11.6,0.5) node[SP] {} ++(0.6,0) node[right]
				{\small processing with delay};
		\else
			\draw[line width=1pt,->] (-3.1,-2) to ++(1.2,0) node[right]
				{\small useful signal};
			\draw[black!30!white,line width=1pt,->] 
				(-3.1,-2.8) to ++(1.2,0) node[right]
				{\small \textcolor{black}{known interference\phantom{g}}};
			\draw[dashed,line width=1pt,->] (4.1,-2) to ++(1.2,0) node[right]
				{\small interference\phantom{g}};
			\draw (4.7,-2.8) node[SP] {} ++(0.6,0) node[right]
				{\small processing with delay};
		\fi
	\end{tikzpicture}
	\caption{Decomposition of the Source Transmit Signal for the PDF Strategy}
	\label{fig:coding-pdf}
\end{figure}
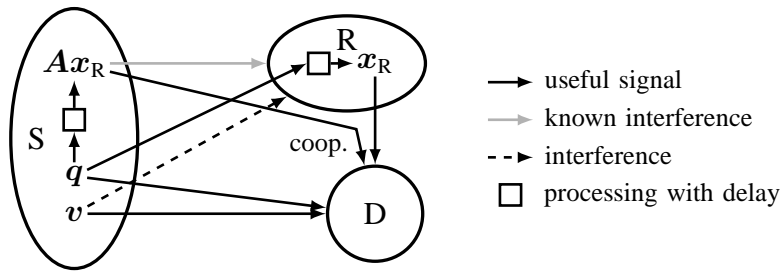
The meaning of these signal parts is illustrated in~Figure~\ref{fig:coding-pdf}:
\begin{itemize}
\item $\q$ contains the new information to be decoded by the relay,
\item $\A \xr$ denotes the cooperative part, which allows the source and the relay to cooperatively transmit the message part $w_{b-1}'$, which the relay has previously decoded, to the destination, and
\item $\v$ represents the new information not to be decoded by the relay, i.e., the part of the source message $w_{b}$ that is conveyed to the destination over the direct link only.
\end{itemize}
Note also that since $w_{b}''$ is not supposed to be decoded by the relay, $\v$ acts as interference at the relay. The matrix $\S = \Cq + \Cv$, which is defined in the primal decomposition approach, can hence be thought of as the covariance of the \emph{innovative part} of the source signal, whereas $\Cs - \S = \A\Cr\A^{\he}$ is the covariance of the \emph{cooperative part}. 

\subsection{Evaluation of $\rpdf$}
If one actually wants to evaluate $\rpdf$, it is not convenient to express the correlation of $\xs$ and $\xr$ by means of $\A$ and $\Cr$ because the corresponding maximization problem would then contain the product $\A \Cr \A^\he$ of two optimization variables. However, this issue can be avoided as follows. Instead of decomposing the source input into three independent signal parts, let
\begin{align}
	\xs = \u + \v
\end{align}
with $\v \sim \CN(\zeros,\Cv)$ being independent of both $\u \sim \CN(\zeros,\Cu)$ and $\xr \sim \CN(\zeros,\Cr)$, cf.~\cite{Gerdes13Zero-ForcingPartialDecode-and-Forward, Weiland13PartialDecode-and-ForwardRates}. Moreover, let $\check{\C}$ denote the joint covariance matrix of $\u$ and $\xr$, i.e., $\left[ \begin{smallmatrix} \u \\ \xr \end{smallmatrix} \right] \sim \CN(\zeros,\check{\C})$ with
\begin{align}
	\check{\C} = \begin{bmatrix}
		\Cu & \Cur \\ \Cur^{\he} & \Cr
	\end{bmatrix}.
\end{align}
Then, the correlation between $\xs$ and $\xr$ is specified by the cross-covariance matrix $\Cur$ and
\begin{align}
	I(\u; \yr | \xr) &= \log \frac{|\id + \Hsr(\Cugr + \Cv)\Hsr^{\he}|}{|\id + \Hsr \Cv \Hsr^{\he}|},
\end{align}
where $\Cugr = \Cu - \Cur \Cr^{+} \Cur^{\he}$ denotes the conditional covariance matrix of $\u$ given $\xr$. By introducing an auxiliary variable $\Q = \Cugr \psd \zeros$, relaxing the equality constraint to $\zeros \nsd \Q \nsd \Cugr$, and subsequently applying the Schur complement condition for positive semidefinite matrices~\cite[Appendix~A.5.5]{Boyd2004}, it can eventually be shown that
\begin{multline}\label{eq:rpdf}
\begin{aligned}
	\rpdf = \max_{\Q, \check{\C}, \Cv} \, \min \, \Biggl\{
	&\log \frac{|\id + \Hsr(\Q + \Cv)\Hsr^{\he}|}{|\id + \Hsr \Cv \Hsr^{\he}|} 
	+ \log |\id + \Hsd \Cv \Hsd^{\he}|, \\
	&\log |\id + \Hsd \Cv \Hsd^{\he} + \HSRd \check{\C} \HSRd^{\he}| \Biggr\}
\end{aligned} \\
	\st \quad \Q, \Cv \psd \zeros, \ \check{\C} - \Ds^{\he} \Q \Ds \psd \zeros, \
	\trace(\Cv + \Ds \check{\C} \Ds^{\he}) \leq \PS, \  
	\trace(\Dr \check{\C} \Dr^\he) \leq \PR,
\end{multline}
where
\begin{align}
	\Ds = \left[ \id_{\Ns}, \zeros_{\Ns \times \Nr} \right], \qquad
	\Dr = \left[ \zeros_{\Nr \times \Ns}, \id_{\Nr} \right].
\end{align}

In fact, the PDF rate maximization problem~\eqref{eq:rpdf} was already derived in~\cite{Gerdes13Zero-ForcingPartialDecode-and-Forward, Weiland13PartialDecode-and-ForwardRates}, but when those papers were written, it was not yet clear whether Gaussian channel inputs maximize the achievable PDF rate. Now, on the other hand, we can conclude from Theorem~\ref{thm:general} that~\eqref{eq:rpdf} yields the maximum achievable PDF rate for the Gaussian MIMO relay channel. Note also that~\eqref{eq:rpdf} becomes equivalent to~\eqref{eq:rpdf-achievable} if $\u = \q + \A\xr$ with $\q$ and $\xr$ being independent. In particular, we have $\Cugr = \Cq$ and
\begin{align}\label{eq:Ccheck}
	\check{\C} = \begin{bmatrix}
		\Cq + \A\Cr\A^{\he} & \A\Cr \\
		\Cr\A^{\he} & \Cr \end{bmatrix}
		= \Ds^{\he} \Cq \Ds + \begin{bmatrix} \A \\ \id \end{bmatrix} \Cr
		\begin{bmatrix} \A \\ \id \end{bmatrix}^{\he}
\end{align}
in this case. Furthermore, the auxiliary variable $\Q$ in problem~\eqref{eq:rpdf} is equal to the conditional covariance matrix $\Cugr$ in the optimum, i.e., the optimal auxiliary variable satisfies $\Q = \Cq$ if $\u = \q + \A\xr$. Therefore, the constraint $\check{\C} - \Ds^{\he} \Q \Ds \psd \zeros$ can be replaced by $\Cr \psd \zeros$, and the equivalence of problems~\eqref{eq:rpdf-achievable} and~\eqref{eq:rpdf} then simply follows from plugging~\eqref{eq:Ccheck} into~\eqref{eq:rpdf}.

Unfortunately, problem~\eqref{eq:rpdf} is still nonconvex due to the term $|\id + \Hsr \Cv \Hsr^{\he}|$ in the denominator of the first logarithm, which results form the fact that $\v$ must be considered as interference at the relay. To the best of our knowledge, an algorithm to compute the globally optimal solution of~\eqref{eq:rpdf} for the general case has yet to be derived, but suboptimal solution approaches have already been proposed. The zero-forcing (ZF) approach used in~\cite{Gerdes13Zero-ForcingPartialDecode-and-Forward} is based on canceling the interference the relay would suffer from, i.e.,  the part of the source input the relay is not supposed to decode. Another approach uses the so-called inner approximation algorithm~(IAA)~\cite{Marks78GeneralInnerApproximation}, which solves a sequence of approximating convex optimization problems instead of the original nonconvex one~\cite{Weiland13PartialDecode-and-ForwardRates}. For Rayleigh fading channels, both of these suboptimal PDF schemes can achieve considerable gains compared to DF and approach the CSB if the source is equipped with more antennas than the relay~\cite{Weiland13PartialDecode-and-ForwardRates, Gerdes13Zero-ForcingPartialDecode-and-Forward}. \red{Furthermore, it has been shown in~\cite{Gerdes14OptimalPartialDecode-and-Forwarda} that the ZF scheme yields the maximum achievable PDF rate $\rpdf$ if
\begin{align}
	\rank(\Hsr) + \rank(\Hsd) = \rank( [\Hsr^{\he}, \Hsd^{\he}] ),
\end{align}
i.e., for cases where the row spaces of the channel gain matrices $\Hsr$ and $\Hsd$ are disjoint, cf.~\cite{Marsaglia74EqualitiesandInequalities}.} We remark that this is one of the special cases for which it was already known that Gaussian channel inputs maximize the achievable PDF rate.

An interesting point that has not been examined so far is whether the primal decomposition approach, which turned out to be the key to proving Theorem~\ref{thm:aligned}, is also helpful for deriving an algorithm that solves the general PDF rate maximization problem~\eqref{eq:rpdf}. This question is left open for future research.

\subsection{Generalizations}

Theorems~\ref{thm:aligned} and~\ref{thm:general} remain valid if the source and relay power constraints are replaced by more general constraints such as shaping constraints. More specifically, as long as the constraints only depend on the joint covariance matrix of $\xs$ and $\xr$, we can still invoke the entropy maximizing property of the Gaussian distribution for the enhanced channel in the converse part of Theorem~\ref{thm:aligned} to show that Gaussian channel inputs maximize the achievable PDF rate for the aligned Gaussian MIMO relay channel. Moreover, the generalization of this result from the aligned to the general case does not depend on the constraints.

In addition, Theorems~\ref{thm:aligned} and~\ref{thm:general} can also be generalized to the Gaussian MIMO relay channel with general (proper or improper) complex Gaussian noise. To see this, first note that using the same arguments as for the complex-valued case, it can be proved that Gaussian channel inputs maximize the achievable PDF rate for the real-valued Gaussian MIMO relay channel. Applying this observation to a composite real representation, where the real and imaginary parts of complex signals are stacked in vectors of twice the dimension~\cite{Adali11Complex-ValuedSignalProcessing:}, we directly obtain that general complex Gaussian channel inputs maximize the achievable PDF rate in case of general complex Gaussian noise.

However, when taking the detour over the composite real representation, it is not easy to see whether the optimal channel inputs in the corresponding complex-valued system are proper or improper. This would have to be deduced from the structures of the covariance matrices of the composite real vectors. For the practically important case (cf.~\cite{Neeser93ProperComplexRandom}, for example) of proper complex Gaussian noise, we could use the argumentation from~\cite{Hellings14OptimalGaussianSignaling} to show that the channel inputs should also be proper. We avoided this involved argumentation in the proofs of Theorems~\ref{thm:aligned} and~\ref{thm:general} by restricting our considerations to proper Gaussian noise right from the beginning.

\bibliographystyle{IEEEtran}
\bibliography{IEEEabrv,IEEEconf,../../../Dissertation/dissertation_lege.bib}

\end{document}